\documentclass[a4paper]{article}
\usepackage{graphicx}
\usepackage{authblk}
\usepackage{amsmath}
\usepackage{amsthm}
\usepackage{amssymb}
\usepackage{hyperref}
\usepackage[letterpaper, left=1in, right=1in, bottom=1in, top=1in]{geometry}
\usepackage{algorithm}
\usepackage{enumerate}
\usepackage{lscape}
\usepackage[noend]{algpseudocode}

\def\CTSM{CTMSeq}

\newcommand{\mfi}{\mathsf{mfi}}
\newcommand{\minidx}{\mathsf{minidx}}
\newcommand{\CT}{\mathit{CT}}
\newcommand{\nil}{\mathit{nil}}
\newcommand{\SUBIDX}[2]{\mathcal{I}^{#1}_{#2}}
\newcommand{\RC}{\mathrm{R}}
\newcommand{\LC}{\mathrm{L}}
\newcommand{\subseq}{\sqsubseteq}
\newcommand{\LFI}{\mathit{LFI}}
\newcommand{\new}{\mathit{new}}

\newcommand{\set}[1]{\{\kern0.05em#1\kern0.05em\}}
\newcommand{\sete}[1]{\{\kern0.00em#1\kern0.00em\}}

\DeclareMathOperator*{\argmax}{arg\,max}

\newtheorem{theorem}{Theorem}
\newtheorem{corollary}{Corollary}
\newtheorem{lemma}{Lemma}

\newtheorem{example}{Example}

\theoremstyle{definition}
\newtheorem{definition}{Definition}

\begin{document}
\title{Cartesian Tree Subsequence Matching}
\author[1]{Tsubasa~Oizumi}
\author[1]{Takeshi~Kai}
\author[2]{Takuya~Mieno}
\author[3,4]{Shunsuke~Inenaga}
\author[2]{Hiroki Arimura}
\affil[1]{Graduate School of Information Science and Technology, Hokkaido University\\
  \texttt{oizumi.tsubasa.e2@elms.hokudai.ac.jp}}
\affil[2]{Faculty of Information Science and Technology, Hokkaido University\\
  \texttt{\{takuya.mieno,arim\}@ist.hokudai.ac.jp}}
\affil[3]{Department of Informatics, Kyushu University\\
  \texttt{inenaga@inf.kyushu-u.ac.jp}}
\affil[4]{PRESTO, Japan Science and Technology Agency}
\date{}
\maketitle
\begin{abstract}
  Park~\textit{et al}. [TCS 2020] observed that the similarity between two (numerical) strings can be captured by the Cartesian trees: The Cartesian tree of a string is a binary tree recursively constructed by picking up the smallest value of the string as the root of the tree.
  Two strings of equal length are said to Cartesian-tree match if their Cartesian trees are isomorphic. 
  Park~\textit{et al}. [TCS 2020] introduced the following \emph{Cartesian tree substring matching} (\emph{CTMStr}) problem: Given a text string $T$ of length $n$ and a pattern string of length $m$, find every consecutive substring $S = T[i..j]$ of a text string $T$ such that $S$ and $P$ Cartesian-tree match.
  They showed how to solve this problem in $\tilde{O}(n+m)$ time.
  In this paper, we introduce the \emph{Cartesian tree subsequence matching} (\emph{\CTSM}) problem, 
  that asks to find every minimal substring $S = T[i..j]$ of $T$ such that 
  $S$ contains a subsequence $S'$ which Cartesian-tree matches $P$.
  We prove that the {\CTSM} problem can be solved efficiently, in $O(m n p(n))$ time,
  where $p(n)$ denotes the update/query time for dynamic predecessor queries.
  By using a suitable dynamic predecessor data structure,
  we obtain $O(mn \log \log n)$-time and $O(n \log m)$-space solution for {\CTSM}.
  This contrasts {\CTSM} with closely related \emph{order-preserving subsequence matching} (\emph{OPMSeq})
  which was shown to be NP-hard by Bose \textit{et al}. [IPL 1998].
\end{abstract}

\section{Introduction}
\label{ch:intro}

A time series is a sequence of events which can be represented by symbols or numbers in many cases.
An \emph{episode} is a collection of events which occur in a short time period.
The \emph{episode matching} problem asks to find every \emph{minimal} substring $S = T[i..j]$ of a text $T$ such that a pattern $P$ is a (non-consecutive) subsequence of $S$.
Let $n$ and $m$ be the lengths of the text $T$ and the pattern $P$, respectively.
There exists a na\"ive $O(mn)$-time $O(1)$-space algorithm for episode matching, which scans the text back and forth.
In 1997, Das~\textit{et al}.~\cite{DasFGGK97} presented a weakly subquadratic $O(mn / \log m)$-time $O(m)$-space algorithm for episode matching.
Very recently, Bille~\textit{et al}.~\cite{abs-2108-08613} showed that even a simpler version of episode matching, which computes the \emph{shortest} substring containing $P$ as a subsequence, cannot be solved in strongly subquadratic $O((mn)^{1-\epsilon})$ time for any constant $\epsilon > 0$, unless the Strong Exponential Time Hypothesis (SETH) fails.

In some applications, such as analysis of time series data of stock prices, 
one is often more interested in finding patterns of price fluctuations rather than the exact prices.
The \emph{order preserving matching} (\emph{OPM}) model~\cite{kim:eades:etal:park:tokuyama:tcs2014oppm} is motivated for such purposes, 
where the task is to find consecutive substring $S$ of a numeric text string $T$ such that the relative orders of values in $S$ are the same as that of a query numeric pattern string $P$.
The order preserving \emph{substring} matching problem (OPMStr) can be solved in $\tilde{O}(n+m)$ time~\cite{kim:eades:etal:park:tokuyama:tcs2014oppm,kubica2013linear,ChoNPS15,CrochemoreIKKLP16}.
On the other hand, the order preserving \emph{subsequence} matching problem (OPMSeq) is known to be NP-hard~\cite{bose:ipl1998lis:npc:pattern}.
Another known model of pattern matching, called \emph{parameterized matching} (\emph{PM}), 
is able to capture structures of strings, namely, two strings are said to parameterized match 
if one string can be obtained by applying a character bijection to the other string~\cite{Baker93}.
Again, the parameterized \emph{substring} matching problem (PMStr) can be solved in $\tilde{O}(n+m)$ time (see~\cite{Baker93,Baker96,IduryS96,FujisatoNIBT21,DBLP:journals/dam/MendivelsoTP20} and references therein), 
but the parameterized \emph{subsequence} matching (PMSeq) is NP-hard~\cite{KellerKL09}.
We remark that both order preserving matching and parameterized matching belong to 
a general framework of pattern matching called the \emph{substring-consistent equivalence relation} (\emph{SCER})~\cite{MatsuokaSCER2016}.
Let $\approx$ denote a string equivalence relation,
and suppose that $X \approx Y$ holds for two strings $X$ and $Y$ of equal length $n$.
We say that $\approx$ is an SCER if $X[i..j] \approx Y[i..j]$ hold for any $1 \leq i \leq j \leq n$.

\emph{Cartesian tree matching} (\emph{CTM}), proposed by Park~\textit{et al}.~\cite{ParkBALP20}, is a new class of SCER 
that is also motivated for numeric string processing.
The Cartesian tree $\CT(T)$ of a string $T$ is a binary tree such that 
the root of $\CT(T)$ is $i$ if $i$ is the leftmost occurrence of the smallest value in $T$,
the left child of the root $T[i]$ is $\CT(T[1..i-1])$,
and the right child of the root $T[i]$ is $\CT(T[i+1..n])$.
We say that two strings Cartesian-tree match if the Cartesian trees of the two strings are isomorphic as ordered trees~\cite{hoffmann1982pattern}, i.e.,
preserving both the parent and sibling orders.
Observe that CTM is similar to OPM. For instance, strings $(\mathtt{7}, \mathtt{2}, \mathtt{3}, \mathtt{1}, \mathtt{5})$ and $(\mathtt{9}, \mathtt{2}, \mathtt{4}, \mathtt{1}, \mathtt{6})$
both Cartesian-tree match and order-preserving match.
It is easy to observe that if two strings order-preserving match, then they also Cartesian-tree match, but the opposite is not true in general.
Thus CTM allows for more relaxed pattern matching than OPM.
Indeed, 
the constraints for OPM that impose the relative order of all positions in the pattern
can be too strict for some applications~\cite{ParkBALP20}.
For example, two strings $(\mathtt{7}, \mathtt{2}, \mathtt{3}, \mathtt{1}, \mathtt{5})$ and $(\mathtt{6}, \mathtt{2}, \mathtt{4}, \mathtt{1}, \mathtt{9})$ both having 
a w-like shape do not order-preserving match.
On the other hand, their similarity can be captured with CTM,
since $(\mathtt{7}, \mathtt{2}, \mathtt{3}, \mathtt{1}, \mathtt{5})$ and $(\mathtt{6}, \mathtt{2}, \mathtt{4}, \mathtt{1}, \mathtt{9})$ Cartesian-tree match.
This lead to the study of the Cartesian tree \emph{substring} matching (CTMStr) problem,
which asks to find every substring $S$ of $T$ such that $S$ and $P$ Cartesian-tree match.
The CTMStr problem can be solved efficiently, in $\tilde{O}(n + m)$ time~\cite{ParkBALP20,SongGRFLP21}.

On the other hand, since real-world numeric sequences contain errors and indeterminate values,
patterns of interest may not always appear consecutively in the target data.
Therefore numeric sequence pattern matching scheme, which allows for skipping some data and 
matching to non-consecutive subsequences, is desirable.
However, such pattern matching is not supported by the CTMStr algorithms.
Given the aforementioned background,
this paper introduces Cartesian tree \emph{subsequence} matching (CTMSeq),
and further shows that this problem can be solved efficiently.
Namely, we can find, in time polynomial in $n$ and $m$,
every minimal substring $S = T[i..j]$ of a text $T$ such that
there exists a subsequence $S'$ of $S$ where $\CT(S')$ and $\CT(P)$ are isomorphic.
We remark that this is the CTM version of episode matching, 
which is also the first polynomial-time subsequence matching under SCER 
(except for exact matching, which is episode matching).

The contribution of this paper is the following:
\begin{itemize}
\item We first present a simple algorithm for solving {\CTSM} in $O(mn^2)$ time and $O(mn)$ space based on dynamic programming (Section~\ref{ch:algobasic}, Algorithm~\ref{alg:dp}).
\item We present a faster $O(mn \log \log n)$-time $O(mn)$-space algorithm for solving {\CTSM} (Section~\ref{ch:algofaster}, Algorithm~\ref{alg:vEB}).
To achieve this speed-up, we exploit useful properties of our method
that permits us to improve the $O(n^2)$-time part of Algorithm~\ref{alg:dp} with $O(n)$ predecessor queries.
\item We present space-efficient versions of the above algorithms that require only $O(n \log m)$ space, which are based on the idea from the \emph{heavy-path decomposition} (Section~\ref{ch:hld}).
\end{itemize}

Technically speaking, our algorithms are related to the work by Gawrychowski et al.~\cite{GawrychowskiGL20},
who considered the problem of deciding whether two \emph{indeterminate} strings of equal length $n$ match under SCER.
They showed that the CTM version of the problem can be solved in $O(n \log^2 n)$ time with $O(n \log n)$ space 
when the number $r$ of uncertain characters in the strings is constant, using predecessor queries.
They also proved that the OPM and PM versions of the problem are NP-hard for $r = 2$.
NP-hardness for the OPM version in the case of $r = 3$ was previously shown in~\cite{HenriquesFRB18}.
Our results on {\CTSM} can be seen as yet another example that differentiates between CTM and OPM 
in terms of the time complexity class.

 \section{Preliminaries}
\label{ch:prelim}

\subsection{Basic Notations and Assumptions}
For any positive integers $i, j$ with $1\le i\le j$, we define a set $[i] = \sete{1,\dots,i}$ of integers and
a discrete interval $[i,j] = \set{i, i+1, \ldots, j}$.
Let $\Sigma = \{1, \ldots, \sigma\}$ be an \emph{integer alphabet} of size $\sigma$.
An element of $\Sigma$ is called a \emph{character}.
A sequence of characters is called a \emph{string}.
The \emph{length} of string $S$ is denoted by $|S|$.
The empty string $\varepsilon$ is the string of length $0$.
For a string $S = (S[1], S[2], \ldots, S[|S|])$, $S[i]$ denotes the $i$-th character of $S$ for each $i$ with $1 \le i \le |S|$.
For each $i,j$ with $1 \le i \le j \le |S|$, $S[i.. j]$ denotes the \emph{substring} of $S$ starting from $i$ and ending at $j$.
For convenience, let $S[i.. j] = \varepsilon$ for $i > j$.
We write $\min(S) := \min \{S[i] \mid i \in [n]\}$ for the minimum value contained in the string $S$.
In this paper, all characters in the string $S$ assume to be different from each other without loss of generality~\cite{kim:eades:etal:park:tokuyama:tcs2014oppm}
\footnote{
If the same character occurs more than once in $S$, the pair $ci = (c, i)$ of the original character $c$ and index $i$ can be extended as a new character to satisfy the assumption.
}.
Under the assumption, we denote by $\minidx(S) := i$ the unique index satisfying the condition $S[i] = \min(S)$.
For any $0\le m\le n$, let $\SUBIDX nm$ be the set consisting of
all
\textit{subscript sequence} $I=(i_1, \ldots, i_m) \in [n]^m$ in ascending order satisfying $1 \leq i_1 < \cdots < i_m \leq n$.
Clearly, $|\SUBIDX nm| = \binom{n}{m}$ holds.
For a subscript sequence $I=(i_1, \ldots, i_m) \in \SUBIDX nm$, we denote by $S_I := (S[i_1], \ldots, S[i_m])$ the \emph{subsequence} of $S$ corresponding to $I$.
Intuitively, a subsequence of $S$ is a string obtained by removing zero or more characters from $S$ and concatenating the remaining characters without changing the order.
For a subscript sequence $I=(i_1, \ldots, i_m) \in \SUBIDX nm$ and its elements $i_s, i_t \in I$ with $i_s\le i_t$, $I[i_s:i_t]$ denotes the substring of $I$ that starts with $i_s$ and ends with $i_t$.
In this paper, we assume the standard \emph{word RAM model} of word size $w = \Omega(\log n)$.
Also we assume that $\sigma \le 2^w$, i.e., any character in $\Sigma$ fits within a single word.

\subsection{Cartesian Tree}

\begin{figure}[tb]
\begin{center}
\includegraphics[scale=0.70]{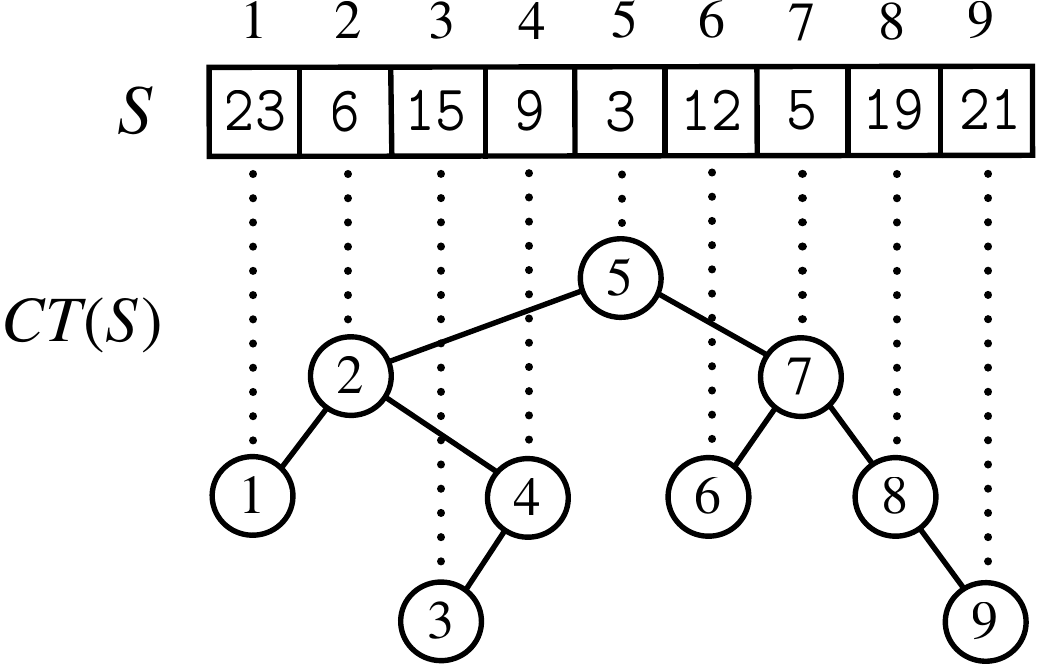}
\end{center}
\caption{
  Illustration for Cartesian tree $\CT(S)$ of $S = (\texttt{23}, \texttt{6}, \texttt{15}, \texttt{9}, \texttt{3}, \texttt{12}, \texttt{5}, \texttt{19}, \texttt{21})$.
  Since the minimum value among $S$ is $S[5]$, node $v = 5$ is the root of $\CT(S)$, $\CT(S[1..4])$ is the left subtree of $v$, and $\CT(S[6..9])$ is the right subtree of $v$.
  Then, $v.\LC = 2$, $v.\RC = 7$, $S_v = S[1..9] = S$, $S_{v.\LC} = S[1..4]$, and $S_{v.\RC} = S[6..9]$.
}
\label{fig:cartesian_tree}
\end{figure}

The Cartesian tree of string $S$, denoted by $\CT(S)$, is the ordered binary tree recursively defined as follows:
If $S = \varepsilon$, then $\CT(S)$ is empty, and otherwise,
$\CT(S)$ is the tree rooted at $v$ such that
the left subtree of $v$ is $CT(S[1.. v-1])$, and
the right subtree of $v$ is $CT(S[v+1.. |S|])$,
where $v = \minidx(S)$.
For a node $v$,
we denote by $v.\LC$ the left child of $v$ if such a child exists and let $v.\LC = \nil$ otherwise. Similarly, we use the notation $v.\RC$ for the right child of $v$.
$CT(S)_v$ denotes the subtree of $CT(S)$ rooted at $v$. 
We say that two Cartesian trees $CT(S)$ and $CT(S')$ are \textit{isomorphic} as ordered trees~\cite{hoffmann1982pattern}, denoted  $CT(S) = CT(S')$.

There is an interplay between a sequence and its Cartesian tree as follows:
We note that the indices of $S$ identify the nodes of $CT(S)$, and \textit{vice versa}. 
For any node $v$ of $CT(S)$, we define the substring $S_v$ of $S$ recursively as follows: 
\begin{enumerate}[(i)]
\item If $v$ is the root of $CT(S)$, then $S_v = S = S[1..|S|]$. 
\item If $v$ is a node with substring $S_v = S[\ell..r]$, then $S[v]$ is the minimum value in $S[\ell.. r]$, $S_{v.\LC} = S[\ell..v-1]$, and $S_{v.\RC} = S[v+1..r]$.
\end{enumerate}
An example of a Cartesian tree is shown in Figure~\ref{fig:cartesian_tree}.

\subsection{Cartesian Tree Subsequence Matching}
Let $T$ be a \emph{text} string of length $n$ and $P$ be a \emph{pattern} string of length $m \le n$. 
We say that a pattern $P$ \textit{matches} text $T$, denoted by $P \subseq T$, if there exists a subscript sequence 
 $I=(i_1, \ldots, i_m) \in \SUBIDX nm$
of $T$
such that $CT(T_I) = CT(P)$ holds. 
Then, we refer to the subscript sequence $I$ as a \textit{trace}.

A possible choice of the notion of occurrences of a pattern $P$ in $T$ is to employ the traces of $P$ as occurrences. 
However, it is not adequate since there can be exactly $\binom{n}{m}$ traces~\footnote{
which can be achieved by monotone sequences for $P$ and $T$. 
} for a text and a pattern of lengths $n$ and $m$. 
Instead, we employ \emph{minimal occurrence intervals} as occurrences defined as follows.

\begin{definition}[minimal occurrence interval]\rm 
For a text $T[1..n]$ and $P[1..m]$, an interval $[\ell, r] \subseteq [n]$ is said to be an \emph{occurrence interval for pattern $P$ over text $T$} if $P \subseq T[\ell..r]$ holds. 
It is said to be \emph{minimal} if 
there is no occurrence interval $[\ell', r']$ for $P$ over $T$ such that $[\ell^{\prime}, r^{\prime}] \subsetneq [\ell, r]$.
\end{definition}
\begin{example}
\label{example:match}
Let text $T = (\mathtt{11}, \mathtt{3}, \mathtt{8}, \mathtt{6}, \mathtt{16}, \mathtt{19}, \mathtt{5}, \mathtt{15}, \mathtt{21}, \mathtt{24})$ and pattern $P = (\mathtt{9}, \mathtt{2}, \mathtt{17}, \mathtt{4}, \mathtt{13})$.
The occurrence interval $[3, 9]$ for $P$ over $T$ is minimal 
since $I=(3, 4, 6, 8, 9)$ is a trace with $\CT(T_I) = \CT(P)$, 
and there is no other occurrence interval $[\ell, r] \subsetneq [3, 9]$ for $P$ over $T$. 
The interval $[1, 8]$ is an occurrence interval, however, it is not minimal since there is another (minimal) occurrence interval $[1, 5] \subsetneq [1, 8]$ for $P$ over $T$.
Overall, all minimal occurrence intervals for $P$ over $T$ are $[1, 5]$ and $[3, 9]$.
\end{example}
From the definition,
there are $O(n^2)$ occurrence intervals for $P$ over $T$,
while there are $O(n)$ minimal occurrence intervals.
If we have the set of all minimal occurrence intervals, we can easily enumerate all occurrence intervals in constant time per occurrence interval.
Thus, we focus on minimal occurrences in this paper. 
Now, the main problem of this paper is formalized as follows:

\begin{definition}[Cartesian Tree Subsequence Matching ({\CTSM})]\rm
Given two strings $T[1..n]$ and $P[1..m]$, find all minimal occurrence intervals for $P$ over~$T$.
\end{definition}

We can easily see that {\CTSM} can be solved in $O(m \binom{n}{m})$ time by simply enumerating all possible subscript sequences. 
However, its time complexities are too large to apply to real-world data sets. 
Hence, our goal here is to devise efficient algorithms running in polynomial time.

In the rest of this paper, we fix text $T$ of arbitrary length $n$ and pattern $P$ of arbitrary length $m$ with $0 < m \le n$.

 \section{\texorpdfstring{$O(mn^2)$}{O(mn2)}-time Dynamic Programming Algorithm}
\label{ch:algobasic}
This section describes an algorithm based on dynamic programming which runs in time $O(mn^2)$.
We later improve the running time to $O(mn \log \log n)$ in Section~\ref{ch:algofaster}.

\subsection{A Simple Algorithm}
By dynamic programming approach, we can obtain a simple algorithm for {\CTSM} with $O(mn^3)$ time and $O(mn^2)$ space complexities as follows. 
It recursively decides if the substring $P_v$ matches in $T[\ell..r]$ for all indices $v$ of $P$ and all intervals $[\ell.. r]$ in $T$ from shorter to larger. 
These complexities mainly come from 
that it iterates the loop for $O(n^2)$ possible intervals in $T$. 
In the following section, we devise more efficient algorithms in time and space complexities by introducing the notion of \emph{minimal fixed-intervals}. 

\subsection{Minimal Fixed-interval}
To solve {\CTSM} without iterating for all possible intervals, we focus on fixing the corresponding locations between node $v$ of $CT(P)$ and index $i$ of $T$.
For a node $v \in [m]$ and index $i \in [n]$,
we refer to a pair $(v, i)$ as a $pivot$.
Then,
we define the minimal interval fixed with pivot $(v, i)$,
called the \emph{minimal fixed-interval}.
\begin{definition}[(minimal) fixed-interval]\rm
\label{def:kukan}
For pivot $(v, i) \in [m] \times [n]$,
interval $[\ell, r] \subseteq [n]$ is called a \emph{fixed-interval with the pivot $(v, i)$} if there exists a trace $I = (i_1, \ldots, i_{|P_v|}) \in \SUBIDX{n}{|P_v|}$
satisfying the following conditions (i)--(iv):
(i) $i$ is an element of $I$,
(ii) $[i_1, i_{|P_v|}] \subseteq [\ell, r]$
(iii) $CT(T_I) = CT(P_v)$ holds, and
(iv) $T[i] = \min(T_I)$ holds. 
Furthermore, a fixed-interval $[\ell, r]$ with the pivot $(v, i)$ is said to be \emph{minimal} if there is no fixed-interval $[\ell^{\prime}, r^{\prime}] \subsetneq [\ell, r]$ with the pivot $(v, i)$ \end{definition}
We show examples of (minimal) fixed-intervals on Figure~\ref{fig:min-range}.
\begin{figure}[tb]
\begin{center}
\includegraphics[scale=0.54]{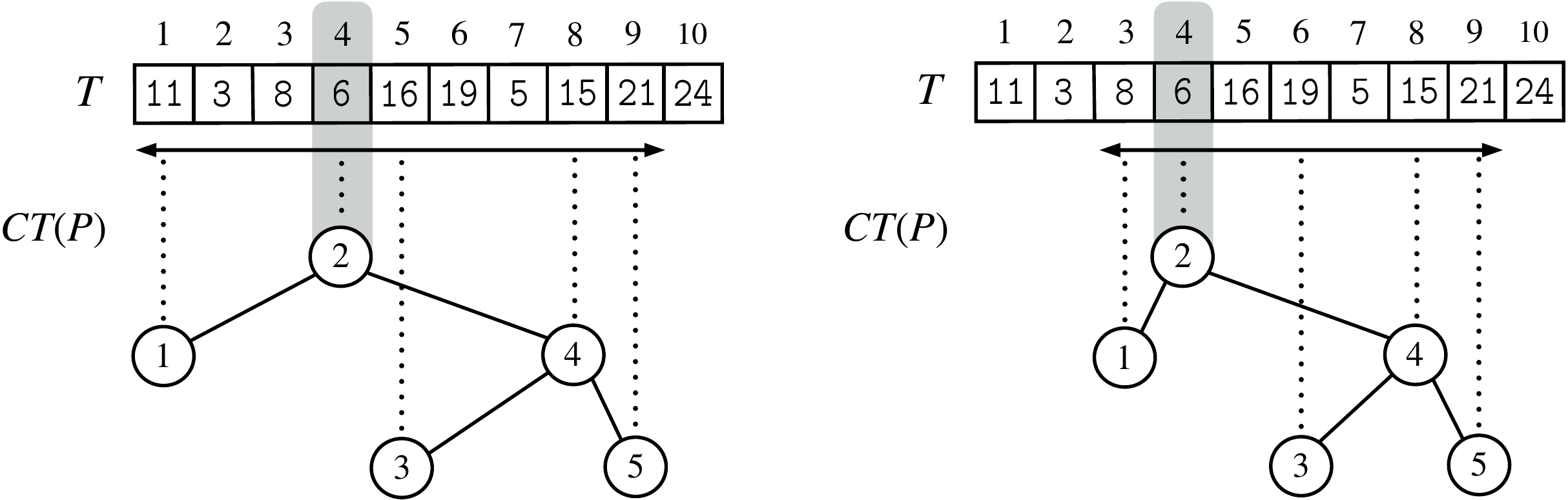}
\end{center}
\caption{
  Illustration for fixed-intervals for the pivot $(v, i)$, where $T=(\mathtt{11}$, $\texttt{3}$, $\texttt{8}$, $\texttt{6}$, $\texttt{16}$, $\texttt{19}$, $\texttt{5}$, $\texttt{15}$, $\texttt{21}$, $\texttt{24})$, $P=(\mathtt{9}$, $\texttt{2}$, $\texttt{17}$, $\texttt{4}$, $\texttt{13})$, and $(v, i) = (2, 4)$.
In the left figure,
  for the trace $I = (1,$ $\textbf{4}$, $5$, $8$, $9)$ indicated by dotted lines,
  the interval $[1, 9]$ is a fixed-interval with the pivot $(v, i)$.
  In the right figure, $[3, 9]$ is a minimal fixed-interval with the pivot $(v, i)$ since there is no fixed-interval $[\ell, r] \subsetneq [3, 9]$ with the pivot $(v, i)$.
}
\label{fig:min-range}
\end{figure}
Here, we give an essential lemma concerning minimal fixed-intervals.
\begin{lemma}
\label{lem:atmost}
  For any pivot $(v, i) \in [m] \times [n]$, there exists at most one minimal fixed-interval with $(v, i)$.
\end{lemma}
\begin{proof}
Assume that there are two minimal fixed-intervals with the pivot $(v, i)$.
Let $[\ell, r]$ and $[\ell^{\prime}, r^{\prime}]$ be two such distinct intervals.
Without loss of generality, assume $\ell\leq \ell^{\prime}$.
Then, by the minimalities of $[\ell, r]$ and $[\ell', r']$, $\ell < \ell'$ and $r < r^{\prime}$ must hold.
From Definition~\ref{def:kukan}, there exist
$I = (\ell, \ldots, i, \ldots, r)$ and $I^{\prime} = (\ell^{\prime}, \ldots, i, \ldots, r^{\prime})$
such that $\CT(T_I) = \CT(T_{I^{\prime}}) = \CT(P_v)$ and $T[i] = \min(T_I) = \min(T_{I^\prime})$.
Since $\CT(T_I) = \CT(T_{I^{\prime}})$ and $T[i] = \min(T_I) =\min(T_{I^{\prime}})$, the right subtree of $i$ in $\CT(T_I)$ is the same as that of $\CT(T_{I^{\prime}})$.
Namely, $\CT(T_{I[i+1:r]}) = \CT(T_{I^{\prime}[i+1:r^{\prime}]})$ holds.
Thus, we have $\CT(T_{I''}) = \CT(P_v)$ where $I''$ is the subscript sequence of length $|I|$
that is the concatenation of $I[\ell^{\prime}:i]$ and $I'[i+1:r]$.
Also, $i \in I''$ and $T[i] = \min(T_{I''})$ hold,
and hence, $[\ell^{\prime}, r]$ is a fixed-interval with the pivot $(v, i)$.
This contradicts that $[\ell^{\prime}, r^{\prime}]$ is a minimal fixed-interval.
\end{proof}
For convenience,
we define the minimal fixed-interval with the pivot $(v, i)$ as $[-\infty, \infty]$ if there is no fixed-interval with the pivot $(v, i)$.
We denote by $\mfi(v, i)$ the minimal fixed-interval with the pivot $(v, i)$.
Let $\mathcal{M} = \{\mfi(\minidx(P), i) \mid i\in[n]\}$ be the set of all the minimal fixed-intervals for the root of $\CT(P)$.
By the definitions of minimal occurrence intervals and minimal fixed-intervals, the next corollary holds:
\begin{corollary}
\label{cor:minimal-interval}
For any minimal occurrence interval $[\ell, r]$ for $P$ over $T$,
$[\ell, r] \in \mathcal{M}$ holds.
Contrary, for any interval $[\ell, r] \in \mathcal{M}$,
if there is no interval $[\ell', r'] \subsetneq [\ell, r]$ such that $[\ell', r'] \in \mathcal{M}$,
$[\ell, r]$ is a minimal occurrence interval for $P$ over $T$.
\end{corollary}
Note that not every intervals $[\ell, r] \in \mathcal{M}$ is a minimal occurrence interval for $P$ over $T$.
We show an example of a interval $[\ell, r] \in \mathcal{M}$ such that $[\ell, r]$ is not a solution of {\CTSM} in Figure~\ref{fig:not-min-range}.

\begin{figure}[tb]
\begin{center}
\includegraphics[scale=0.54]{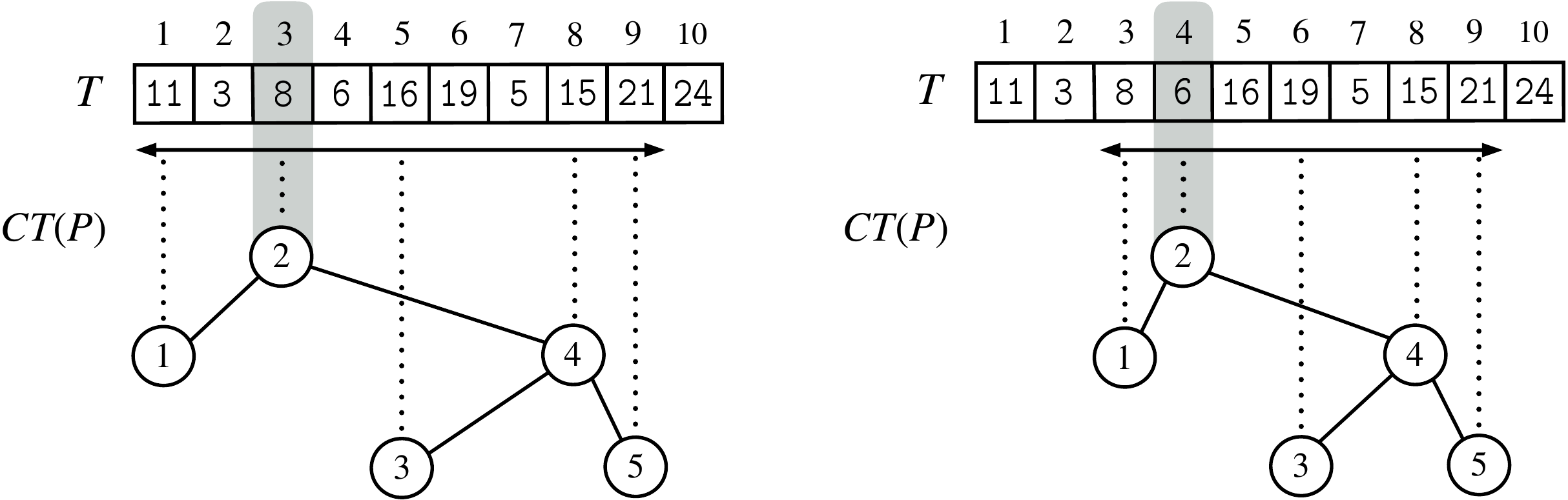}
\end{center}
\caption{
  Illustration for two minimal fixed-intervals, where $T$ and $P$ are the same as in Figure~\ref{fig:min-range}.
  From the figure, $\mfi(2, 3) = [1, 9]$ and $\mfi(2, 4) = [3, 9]$ hold.
  Note that $\mfi(2, 3) = [1, 9]$ is not a solution of {\CTSM} since $\mfi(2, 4) \subsetneq \mfi(2, 3)$ holds.
}
\label{fig:not-min-range}
\end{figure}

\subsection{The Algorithm}
From Corollary~\ref{cor:minimal-interval},
once we compute the set $\mathcal{M}$ of intervals,
we can obtain the solution of {\CTSM} by removing non-minimal intervals from $\mathcal{M}$.
Since every interval in $\mathcal{M}$ except $[-\infty, \infty]$ is a sub-interval of $[1, n]$, we can sort them in $O(n)$ time by using bucket sort, and thus, can also remove non-minimal intervals.

Thus, in what follows, we discuss how to efficiently compute $\mathcal{M}$, i.e., $\mfi(\minidx(P), i)$ for all $i \in [n]$.
Now, we define two functions $L(v, i) = \ell$ and $R(v, i) = r$
for each node $v \in [m]$ in $\CT(P)$ and each index $i \in [n]$,
where $[\ell, r] = \mfi(v, i)$.
Then, our task is, to compute $L(\minidx(P), i)$ and $R(\minidx(P), i)$ for all $i \in [n]$.
Regarding the two functions, we show the following lemma~(see also Figure~\ref{fig:dp} for illustration):
\begin{lemma}\label{lem:recursion}
For any pivot $(v, i) \in [m] \times [n]$,
the following recurrence relations hold:
\begin{equation*}
\label{ali:l-step}
L(v, i) = \left\{
\begin{array}{ll}
    -\infty & \text{\rm{if}~}\mfi(v, i) = [-\infty, \infty],\\
    i & \text{\rm{if}} \ \mfi(v, i) \ne [-\infty, \infty]\\
    & \quad\text{\rm{and} }v.\LC = \nil,\\
    \displaystyle \max_{\substack{1 \leq j \leq i - 1}}{\{L(v.\LC,j) \mid T[i] < T[j], R(v.\LC,j) < i\}} & \text{\rm{otherwise}}.
\end{array}
\right.
\end{equation*}
\begin{equation*}
\label{ali:r-step}
R(v, i) = \left\{
\begin{array}{ll}
    \infty & \text{\rm{if}~}\mfi(v, i) = [-\infty, \infty],\\
i & \text{\rm{if}} \ \mfi(v, i) \ne [-\infty, \infty]\\
    & \quad\text{\rm{and} }v.\RC = \nil,\\
    \displaystyle \min_{\substack{i + 1 \leq j \leq n}} \{R(v.\RC, j) \mid T[i] < T[j], i < L(v.\RC, j) \} & \text{\rm{otherwise}}.
\end{array}
\right.
\end{equation*}
\end{lemma}
\begin{figure}[tb]
\begin{center}
\includegraphics[scale=0.55]{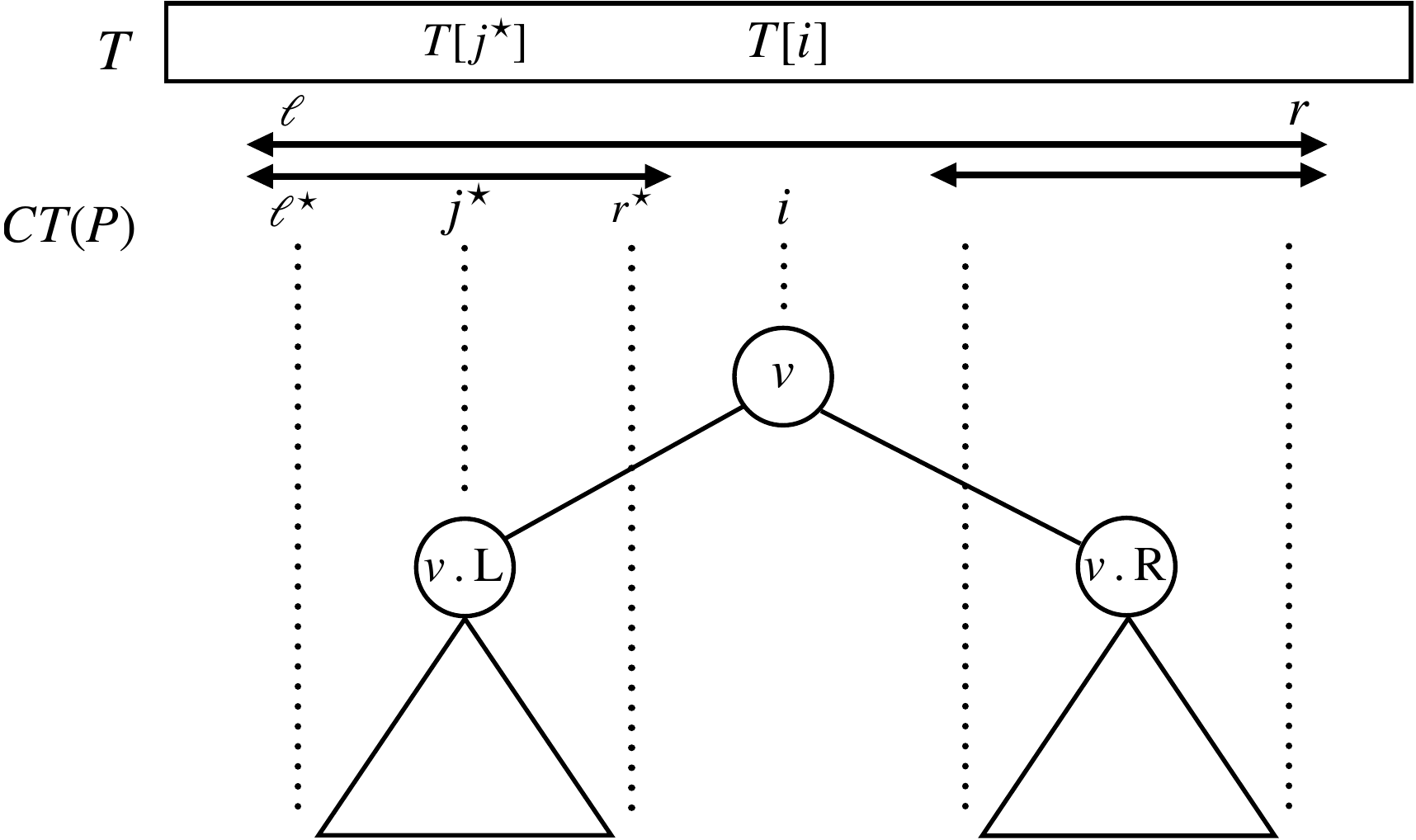}
\end{center}
\caption{
Illustration for intuitive understanding of recurrence relations in Lemma~\ref{lem:recursion}.
The minimal fixed-interval $[\ell, r]$ with the pivot $(v, i)$ can be obtained from $\mfi(v.\LC, j)$ and $\mfi(v.\RC, k)$ by choosing $j$ and $k$ appropriately.
  As for the left subtree of $v$, the candidates for such $j$ must satisfy the conditions that the right-end of $\mfi(v.\LC, j)$ does not exceed $i-1$ and $T[j] > T[i]$.
  To minimize the width of fixed-intervals with $(v, i)$, we choose $j^\star$ that maximizes the left-end of $\mfi(v.\LC, j^\star)$ while satisfying the above conditions.
  Also, symmetric arguments can be applied to the right subtree of $v$.
}
\label{fig:dp}
\end{figure}

\begin{proof}
We prove the validity of the first equation for $L(v, i)$.
The second one can be proven by symmetric arguments.
The first two cases are clearly correct by the definition of minimal fixed-intervals.
We focus on the third case, when $\mfi(v, i) \ne [-\infty, \infty]$ and $v.\LC \ne \nil$.

Let $[\ell, r] = \mfi(v, i)$.
By Definition~\ref{def:kukan}, there exists $I = (\ell, \ldots, i, \ldots, r)$ such that
    $CT(T_I) = CT(P_v)$ and $T[i] = \min(T_I)$.
    We notice that $\CT(P_{v.\LC}) = \CT(T_{I[\ell:\mathit{pre}_i]})$ holds
    where $\mathit{pre}_i$ is the subscript preceding $i$ in $I$.
    Thus, there exists $k$ such that
    $\ell \le k \le i-1$, $T[i] < T[k]$, $R(v.\LC, k) < i$,
    and $L(v.\LC, k) \ge \ell$.
Now, let $j^\star := \argmax_{1 \leq j \leq i - 1} \{L(v.\LC, j) \mid T[i] < T[j], R(v.\LC, j) < i\}$ and
    $[\ell^\star, r^\star] := \mfi({v.\LC}, j^\star)$.
    Then, $\ell^\star = L(v.\LC, j^\star) \ge L(v.\LC, k) \ge \ell$ holds.

For the sake of contradiction, we assume $\ell < \ell^\star$.
    By Definition~\ref{def:kukan},
    there exists $I^\star = (\ell^\star, \ldots, j^\star, \ldots, r^\star)$ such that
    $CT(T_{I^\star}) = CT(P_{v.\LC})$ and $T[j^\star] = \min(T_{I^\star})$.
    Also, by the definition of $j^\star$, $T[i] < T[j^\star]$ and $r^\star < i$ hold.
    Let $I^{\prime}$ be the concatenation of $I^\star$ and $I[i:r]$.
    Note that $I^{\prime} \in \SUBIDX nm$ since $r^\star < i$.
From the above discussions, $\CT(T_{I^\prime}) = \CT(P_v)$ holds
    since $\CT(T_{I^\star}) = \CT(P_{v.\LC})$ and $\min(T_{I^\star}) = T[j^\star] > T[i]$.
Also, $i \in I^\prime$ and $T[i] = \min(T_{I'})$ clearly hold.
Then, by Definition~\ref{def:kukan}, $[\ell^\star, r] \subsetneq [\ell, r]$ is a fixed-interval with $(v, i)$,
    however, this contradicts the minimality of $[\ell, r] = \mfi(v, i)$.
    Therefore, $\ell = \ell^\star$ holds.
    Namely, $L(v, i) = L(v.\LC, j^\star) = \max_{1 \leq j \leq i-1}\{L(v.\LC, j) \mid T[i] < T[j], R(v.\LC, j) < i\}$ holds.
\end{proof}
Algorithm~\ref{alg:dp} is a pseudo code of our algorithm to solve {\CTSM} using dynamic programming based on Lemma~\ref{lem:recursion}.
\begin{algorithm}[tb]
  \caption{Algorithm for solving {\CTSM} using dynamic programming}\label{alg:dp}
  \begin{algorithmic}[1]
\Procedure{cartesian-tree-subsequence-match}{$T[1..n], P[1..m]$}
    \State $L[v][i] \leftarrow -\infty$ {\bf for} all $v \in [m]$ and $i \in [n]$
    \State $R[v][i] \leftarrow \infty$ {\bf for} all $v \in [m]$ and $i \in [n]$
    \State $C \leftarrow \CT(P)$
    \For{each $v \in [m]$ in a bottom-up manner in $C$}
\State call UPDATE-LEFT-MAX($v, T, L, R$)
    \State call UPDATE-RIGHT-MIN($v, T, L, R$)
    \EndFor
    \State enumerate all minimal occurrence intervals for $P$ over $T$ by using $L$ and $R$.
\EndProcedure 
\Function{update-left-max}{$v, T, L, R$}
\If{$v.\LC = \nil$}
    \State $L[v][i] \leftarrow i$ {\bf for} all $i \in [n]$
    \State \Return
    \EndIf
    \For{$i \leftarrow 1$ \textbf{to} $n$}
    \For{$j \leftarrow 1$ \textbf{to} $i - 1$}
    \If{$T[i] < T[j]$ \textbf{and} $R[v.\LC][j] < i$}
    \State $L[v][i] \leftarrow \displaystyle \max (L[v][i], L[v.\LC][j])$
    \EndIf
    \EndFor
    \EndFor
    \EndFunction
    
    \Function{update-right-min}{$v, T, L, R$}
\If{$v.\RC = \nil$}
    \State $R[v][i] \leftarrow i$ {\bf for} all $i \in [n]$
    \State \Return
    \EndIf
    \For{$i \leftarrow 1$ \textbf{to} $n$}
    \For{$j \leftarrow i + 1$ \textbf{to} $n$}
    \If{$T[i] < T[j]$ \textbf{and} $i < L[v.\RC][j]$}
    \State $R[v][i] \leftarrow \displaystyle \min (R[v][i], R[v.\RC][j])$
    \EndIf
    \EndFor
    \EndFor
    \EndFunction
  \end{algorithmic}
\end{algorithm}

\subsubsection*{Correctness of Algorithm~\ref{alg:dp}.}
Algorithm~\ref{alg:dp} computes tables $L[v][i] = L(v, i)$ and $R[v][i] = R(v, i)$ for all pivot $(v, i) \in [m] \times [n]$ in a \emph{bottom-up} manner in $\CT(P)$~(see Line~$5$).
Since the recursion formulae of Lemma~\ref{lem:recursion} hold for every node,
Algorithm~\ref{alg:dp} correctly computes all the minimal fixed-intervals, and thus,
all the minimal occurrence intervals for pattern $P$ over text $T$.

\subsubsection*{Time and Space Complexities of Algorithm~\ref{alg:dp}.}
At Line $4$, we build the Cartesian tree $C$ of a given pattern $P$.
There is a linear-time algorithm to build a Cartesian tree~\cite{gabow1984scaling}, which takes $O(m)$ time here.
In Lines~5--7, we call functions UPDATE-LEFT-MAX and UPDATE-RIGHT-MIN $m$ times since $C$ has $m$ nodes.
It is clear that the functions UPDATE-LEFT-MAX and UPDATE-RIGHT-MIN run in $O(n^2)$ time for each call.
Thus, the total running time of Algorithm~\ref{alg:dp} is $O(mn^2)$.
Also, the space complexity of Algorithm~\ref{alg:dp} is $O(mn)$,
which is dominated by the size of tables $L$ and $R$.

To summarize, we obtain the following theorem:
\begin{theorem}
\label{the:dp}
  The {\CTSM} problem can be solved in $O(mn^2)$ time using $O(mn)$ space.
\end{theorem}

With a few modifications,
we can reconstruct a trace $I = (\ell, \ldots, r) \in \SUBIDX nm$ satisfying $CT(T_I) = CT(P)$ for each minimal occurrence interval $[\ell, r]$.
Precisely, when we compute
the minimal fixed-interval with each pivot $(v, i)$,
we simultaneously compute and store which index will correspond to the root of the left subtree of $v$ fixed at $i$.
We do the same for the right subtree. Using the additional information,
we can reconstruct a desired subscript sequence
by tracing back from the root of $\CT(P)$.
The next corollary follows from the above discussion:
\begin{corollary}
\label{cor:trace}
Once we compute $L(v, i)$ and $R(v, i)$ extended with the information of tracing back for all pivots $(v, i) \in [m] \times [n]$,
we can find a trace $I = (\ell, \ldots, r)$ satisfying $CT(T_I) = CT(P)$ for each minimal occurrence interval $[\ell, r]$ for $P$ over $T$ in $O(m)$ time using $O(mn)$ space.
\end{corollary}

 \section{Reducing Time to \texorpdfstring{$O(mn \log \log n)$}{O(mn log log n)} with Predecessor Dictionaries}
\label{ch:algofaster}

This section describes how to improve the time complexity of Algorithm~\ref{alg:dp} to $O(mn \log \log n)$.
In Algorithm~\ref{alg:dp}, functions UPDATE-LEFT-MAX and UPDATE-RIGHT-MIN require $O(n^2)$ time for each call, which is a bottle-neck of Algorithm~\ref{alg:dp}.
By devising the update order of tables $L(v, i)$ and $R(v, i)$ and using a predecessor dictionary, we improve the running time of the above two functions to $O(n\log\log n)$.

\subsection{Main Idea for Reducing Time}
For any pivot $(v, i) \in [m] \times [n]$,
let $\LFI(v, i) = \{[L(v.\LC, j), R(v.\LC, j)] \mid 1\le j \le n, T[i] < T[j]\}$ be a set of intervals which are candidates for a component of the minimal fixed-interval with $(v, i)$.
By Lemma~\ref{lem:recursion},
$L(v, i) = \max (\{\ell \mid [\ell, r] \in \LFI(v, i), r < i\}\cup\{-\infty\})$
holds if $v.\LC \ne \nil$.
Then, the next observations follow by the definitions:
\begin{itemize}
  \item $\LFI(v, i_1) \subseteq \LFI(v, i_2)$ holds for any $i_1, i_2$ with $T[i_1] > T[i_2]$.
  \item If there are intervals $[\ell_1, r_1], [\ell_2, r_2] \in \LFI(v, i)$ such that $\ell_2 = \ell_1 \le r_1 < r_2$, then
  we can always choose $\ell_1$ as $L(v, i)$.
  \item If there are intervals $[\ell_1, r_1], [\ell_2, r_2] \in \LFI(v, i)$ such that $\ell_2 < \ell_1 \le r_1 \le r_2$, then $\ell_2$ is never chosen as $L(v, i)$.
\end{itemize}
The intuitive explanation of the third observation is shown in Figure~\ref{fig:second-observation}.
\begin{figure}[tb]
  \begin{center}
  \includegraphics[scale=0.45]{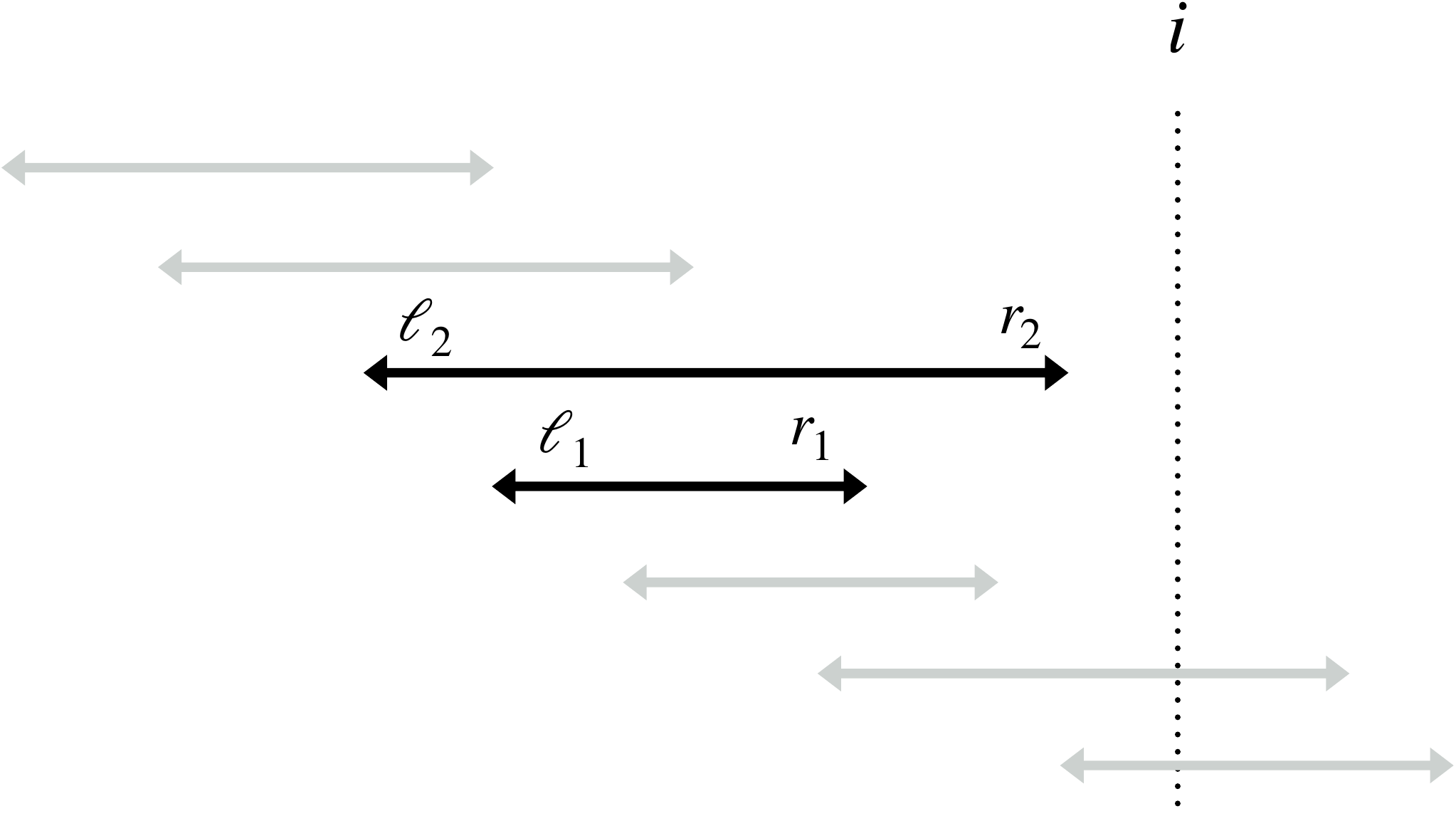}
  \caption{
Illustration for the third observation for $\LFI(v, i)$.
  The double-headed arrows represent the intervals in $\LFI(v, i)$.
  The two intervals $[\ell_1, r_1]$ and $[\ell_2, r_2]$ are in $\LFI(v, i)$ and $[\ell_1, r_1] \subsetneq [\ell_2, r_2]$ holds.
  It is clear that $\ell_2$ is never chosen as $L(v, i)$ for any $i \in [n]$.
  }
  \label{fig:second-observation}
  \end{center}
\end{figure}
From the third observation, we define a subset $\LFI'(v, i)$ of $\LFI(v, i)$, whose conditions are sufficient to our purpose:
Let $\LFI'(v, i)$ be the set of all intervals that are minimal within $\LFI(v, i)$. Namely,
$\LFI'(v, i) = \{[\ell, r] \in \LFI(v, i)\mid$
there is no other interval $[\ell', r'] \in \LFI(v, i)$
such that $[\ell', r']\subsetneq [\ell, r] \}$.
By the third observation,
\begin{equation}\label{eq:L}
  L(v, i) = \max (\{\ell \mid [\ell, r] \in \LFI'(v, i), r < i\}\cup\{-\infty\})
\end{equation}
holds if $v.\LC \ne \nil$.

The main idea of our algorithm is to maintain a set $\mathcal{S}_v$ of intervals so that it satisfies the invariant $\mathcal{S}_v = \LFI'(v, i)$.
To maintain $\mathcal{S}_v$ efficiently, we utilize a data structure called \emph{predecessor dictionary} for $\mathcal{S}_v$ supporting the following operations:
\begin{itemize}
  \item $\texttt{insert}(\mathcal{S}_v, \ell, r)$: insert interval $[\ell, r]$ into $\mathcal{S}_v$,
  \item $\texttt{delete}(\mathcal{S}_v, \ell, r)$: delete interval $[\ell, r]$ from $\mathcal{S}_v$,
  \item $\texttt{pred}(\mathcal{S}_v, x)$: return the interval $[\ell, r] \in \mathcal{S}_v$ on which $r$ is the largest among those satisfying $r < x$ (if it does not exist return $\nil$), and
  \item $\texttt{succ}(\mathcal{S}_v, x)$: return the interval $[\ell, r] \in \mathcal{S}_v$ on which $r$ is the smallest among those satisfying $x < r$ (if it does not exist return $\nil$).
\end{itemize}
To implement a predecessor dictionary for $\mathcal{S}_v$,
we use a famous data structure called \emph{van Emde Boas tree}~\cite{van1977preserving} that performs the operations as mentioned above in $O(\log\log n)$ time each\footnote{The van Emde Boas tree is a data structure for the set of integers, however, it can be easily applied to the set of pairs of integers by associating the first element with the second element.}.
In general, the space usage of van Emde Boas tree is $O(U)$, where $U$ is the maximum of the integers to store.
However, $U = n$ holds in our problem setting, and hence, the space complexity is $O(n)$.

\subsection{Faster Algorithm}

Algorithm~\ref{alg:vEB} shows a function UPDATE-LEFT-MAX that computes $L(v, i)$ for all $i \in [n]$ based on the above idea.
This function can be used to replace the function of the same name in Algorithm~\ref{alg:dp}.
The implementation of function UPDATE-RIGHT-MIN is symmetric.
\begin{algorithm}[tb]
  \caption{Faster algorithm for UPDATE-LEFT-MAX using van Emde Boas tree}\label{alg:vEB}
  \begin{algorithmic}[1]
    \Function{UPDATE-LEFT-MAX}{$v, T, L, R$}  
    \If{$v.\LC = \nil$}
    \State $L[v][i] \leftarrow i$ {\bf for} all $i \in [n]$
    \State \Return
    \EndIf
    \State $\mathcal{S}_v \leftarrow \emptyset$.
    \For{each $i \in [n]$ in the descending order of its value $T[i]$}
    \State $[\ell, r] \leftarrow \texttt{pred}(\mathcal{S}_v, i)$
    \If{$[\ell, r] = \nil$}
    \State $L[v][i] \leftarrow -\infty$
    \State \textbf{continue}
    \EndIf
    \State $L[v][i] \leftarrow \ell$
    \State $\ell_\new \leftarrow L[v.\LC][i]$, $r_\new \leftarrow R[v.\LC][i]$
    \Loop
    \Comment{delete all intervals that become non-minimal}
    \State $[\ell_s, r_s] \leftarrow \texttt{succ}(\mathcal{S}_v, r_\new - 1)$
    \If{$[\ell_s, r_s] = nil$ \textbf{or} $[\ell_\new, r_\new] \not\subseteq [\ell_s, r_s]$}
    \State \textbf{break}
    \EndIf
    \State $\texttt{delete}(\mathcal{S}_v, \ell_s, r_s)$
    \EndLoop
    \State $[\ell_p, r_p] \leftarrow \texttt{pred}(\mathcal{S}_v, r_\new + 1)$
    \If{$[\ell_p, r_p] = \nil$ \textbf{or} $[\ell_p, r_p] \not\subseteq [\ell_\new, r_\new]$}
    \Comment{insert new interval if it is minimal}
    \State $\texttt{insert}(\mathcal{S}_v, \ell_\new, r_\new)$
    \EndIf
    \EndFor
    \EndFunction
  \end{algorithmic}
\end{algorithm}

\subsubsection*{Correctness of Algorithm~\ref{alg:vEB}.}
Remark that $v$ is fixed in Algorithm~\ref{alg:vEB}.
Let $(i_1, \ldots, i_n)$ be the permutation of $[n]$ that is sorted in the order
in which they are picked up by the for-loop at Line 6.
We assume that the invariant $\mathcal{S}_v = \LFI'(v, i_j)$ holds
at the beginning of the $j$-th step of the for-loop. 
The value of $L[v][i_j]$ is determined at either Line 3, 9, or 11.
By Lemma~\ref{lem:recursion}, $L[v][i_j] = L(v,i_j)$ holds
if the value determined at Line 3 or 9.
By the invariant $\mathcal{S}_v = \LFI'(v, i_j)$ and Equation~\ref{eq:L},
$L[v][i_j] = L(v,i_j)$ also holds
if the value determined at Line 11.
Thus, $L(v,i_j)$ is computed correctly.

Next, let us consider the invariant for $\mathcal{S}_v$.
At Line 12, we set $[\ell_\new, r_\new]$ the minimal fixed-interval with $(v.\LC, i_j)$.
In the internal loop at Lines~13--17,
we delete all intervals $[\ell_s, r_s]$ from $\mathcal{S}_v$ such that
$[\ell_s, r_s]$ becomes \emph{non-minimal} within $\mathcal{S}_v\cup\{[\ell_{\new}, r_{\new}]\}$.
To do so, we repeatedly query $\texttt{succ}(\mathcal{S}_v, r_{\new}-1)$
and check whether the obtained interval includes $[\ell_\new, r_\new]$.
Finally, at the last two lines,
we insert the new interval $[\ell_\new, r_\new]$
if it does not include any other interval in $\mathcal{S}_v$.
Then, any intervals in $\mathcal{S}_v$ are not nested each other, and thus,
the invariant $\mathcal{S}_v = \LFI'(v, i_{j+1})$ holds at the end of the $j$-th step.

\subsubsection*{Time and Space Complexities of Algorithm~\ref{alg:vEB}.}
We analyze the number of calls for each operation on a predecessor dictionary.
Firstly, since $\texttt{insert}$ is called only at Line $20$,
it is called at most $n$ times throughout Algorithm~\ref{alg:vEB}.
Similarly, $\texttt{pred}$ at Line 7 and Line 18 is also called $O(n)$ times.
From Line $13$ to Line $17$, $\texttt{succ}$ and $\texttt{delete}$ are called in the internal loop.
The number of calls for $\texttt{delete}$ is at most that of $\texttt{insert}$, and hence, $\texttt{delete}$ is called at most $n$ times, and $\texttt{succ}$ as well.
Thus, throughout Algorithm~\ref{alg:vEB}, the total number of calls for all queries is $O(n)$.
Therefore, the running time of Algorithm~\ref{alg:vEB} is $O(n \log \log n)$.
Also, the space complexity of Algorithm~\ref{alg:vEB} is $O(n)$.

To summarize this section, we obtain the following lemma:
\begin{lemma}
\label{lem:vED}
  Algorithm~\ref{alg:vEB} computes function UPDATE-LEFT-MAX in $O(n \log \log n)$ time using $O(n)$ space.
\end{lemma} \section{Reducing Space to \texorpdfstring{$O(n \log m)$}{O(n log m)}}
\label{ch:hld}
This section describes how to reduce the space complexity of
our algorithm to $O(n \log m)$.
Having the tables $L[v][i]$ and $R[v][i]$ for all pivot $(v, i) \in [m] \times [n]$ requires $\Theta(mn)$ space.
By Lemma~\ref{lem:recursion}, to compute the table values for node $v \in \CT(P)$, we only need the table values for $v.\LC$ and $v.\RC$.
Thus, we can discard the remaining values no longer referenced.
However, even if we discard such unnecessary ones, the space complexity will not be improved in the worst case if we fix the order in which subtree is visited first:
Let us assume that the left subtree is always visited first,
and consider pattern 
\begin{equation}
    P = (k+1, 1, \ldots, k+i, i, \ldots, 2k, k, 2k+1) \label{eq:worst-case-pattern}
\end{equation}
of length $m = 2k+1$.
It can be seen that every non-leaf node in $\CT(P)$ has exactly two children, and the left child is a leaf~(see also Figure~\ref{fig:hld} for a concrete example).
Thus, when we process the node $v$ numbered with $2k$,
we need to store at least $k+1$ tables
since all tables for $k+1$ leaves have been created and not been discarded yet, and it yields $\Theta(mn)$ space.
\begin{figure}[tb]
\begin{center}
\includegraphics[scale=0.55]{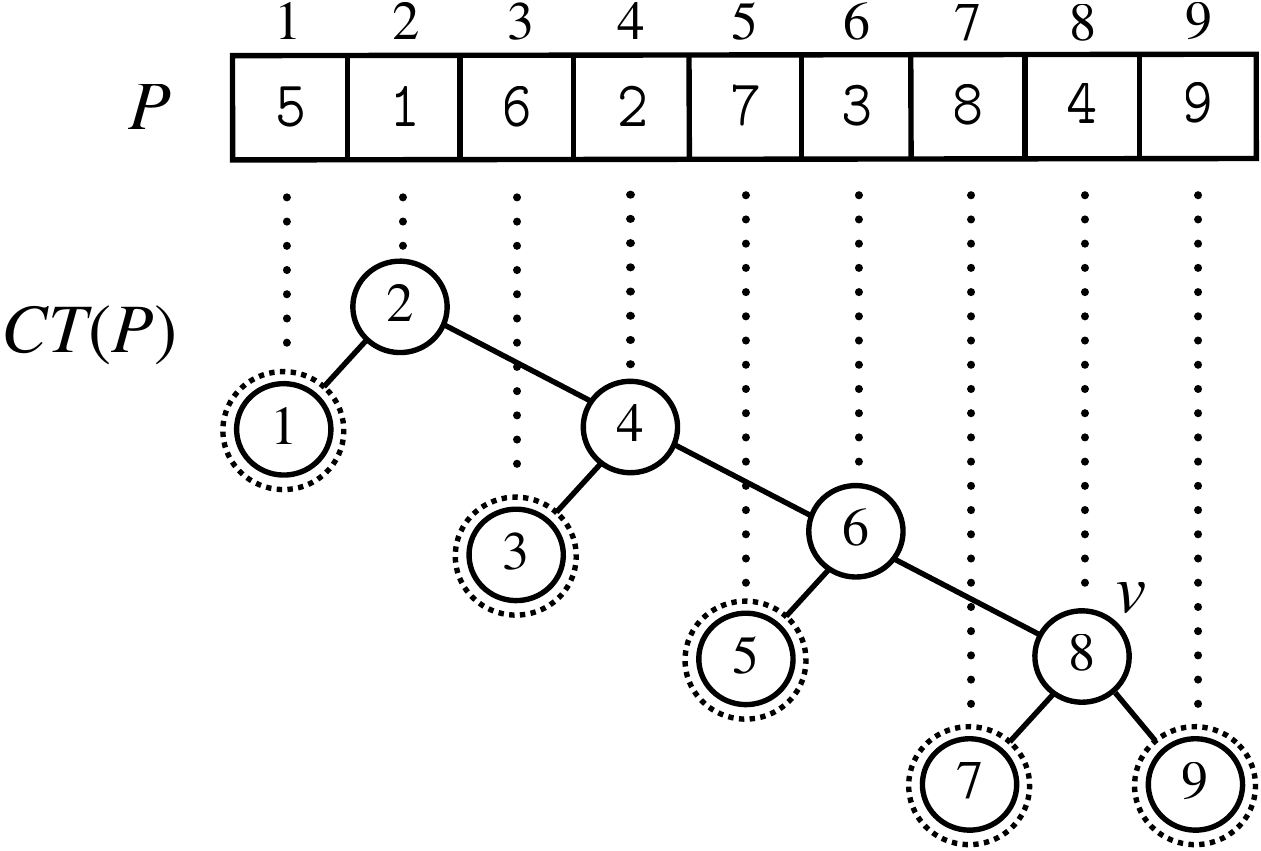}
\end{center}
\caption{
  Illustration for a worst case example of $CT(P)$ which causes the space complexity to be $\Theta(mn)$, where $P = (\mathtt{5}, \mathtt{1}, \mathtt{6}, \mathtt{2}, \mathtt{7}, \mathtt{3}, \mathtt{8}, \mathtt{4}, \mathtt{9})$.
  To compute the values of $L(v, i)$ and $R(v, i)$ for $v = 8$, we only need the values of $L(u, i)$ and $R(u, i)$ for node $u \in \{7, 9\}$.
  However,
  we have not finished computing $L(u, i)$ and $R(u, i)$ for node $u \in \{2, 4, 6\}$ yet,
  so we have to remember all of the values of $L(u, i)$ and $R(u, i)$ for node $u \in \{1, 3, 5, 7, 9\}$ simultaneously.
}
\label{fig:hld}
\end{figure}

To avoid such a case, we add a new rule for which subtree is visited first; when we perform a depth-first traversal, we visit the \emph{larger} subtree first if the current node $v$ has two children.
Specifically, we visit the left subtree first
if $|CT(P)_{v.\LC}| > |CT(P)_{v.\RC}|$,
and visit the right subtree first otherwise,
where the cardinality of a tree means the number of nodes in the tree.
Clearly, the correctness of the modified algorithm relies on the original one (i.e., Algorithm~\ref{alg:dp})
since the only difference is the rule that decides the order to visit.

In the following, we show that the rule makes the space complexity $O(n \log m)$.
We utilize a technique called \emph{heavy-path decomposition}~\cite{harel1984fast}~(a.k.a. heavy-light decomposition).
For each internal node $v \in [m]$ in $\CT(P)$,
we choose one of $v$'s children with the larger subtree size and mark it as \emph{heavy},
and we mark the other one as \emph{light} if it exists.
Exceptionally, we mark the root of $\CT(P)$ as heavy.
Then, it is known that the number of light nodes on any root-to-leaf path is $O(\log m)$~\cite{harel1984fast}.

Now, we prove that the algorithm requires $O(n \log m)$ space at any step.
Suppose we are now on node $u \in [m]$.
Let $\mathsf{p}_u$ be the path from the root to $u$ in $\CT(P)$.
Note that each node $v$ on $\mathsf{p}_u$ is marked as either heavy or light.
For each light node $v_\ell$ on $\mathsf{p}_u$, we have not discarded arrays $L$ and $R$ of size $O(n)$ associated with the sibling of $v_\ell$ to process the parent of $v_\ell$ in a later step.
For each heavy node $v_h$ on $\mathsf{p}_u$, we do not have to remember any array
since we recurse on $v_h$ first, and hence we require only $O(1)$ space for $v_h$.
Since there are at most $O(\log m)$ light nodes on $\mathsf{p}_u$,
the algorithm requires $O(n \log m)$ space at any step.

By combining these discussion with Theorem~\ref{the:dp} and Lemma~\ref{lem:vED}, we obtain our main theorem:
\begin{theorem}\rm
\label{the:set}
  The {\CTSM} problem can be solved in $O(mn \log\log n)$ time
  using $O(n \log m)$ space.
\end{theorem}

Note that the same method as for Corollary~\ref{cor:trace} can not be applied to
the algorithm in this section
since most tables are discarded to save space. \section{Preliminary Experiments}

This section aims to investigate the behavior of each algorithm using artificial data.
In the first experiment we use randomly generated strings to see how the algorithms would behave on average~(Table~\ref{tab:seq-cons-time}).
In the second experiment, we use the worst-case instance presented in Section~\ref{ch:hld} to check the worst-case behavior of the proposed algorithms~(Table~\ref{tab:seq-cons-time-worst-case}).

We conducted experiments on mac OS Mojava 10.14.6 with Intel(R) Core(TM) i5-7360U CPU @ 2.30GHz. 
For each test, we use a single thread and limit the maximum run time by 60 minutes.
All programs are implemented using C++ language compiled with Apple LLVM version 10.0.1 (clang-1001.0.46.4) with -O3 optimization option.
We compared the running time and memory usage of our four proposed algorithms below by varying the length $n$ of text and the length $m$ of pattern:
\begin{itemize}
    \item \texttt{basic}: $O(mn^2)$-time and $O(mn)$-space algorithm (Algorithm~\ref{alg:dp}) explained in Section~\ref{ch:algobasic},
    \item \texttt{basic-HL}: $O(mn^2)$-time and $O(n \log m)$-space algorithm obtained by applying the idea of memory reduction in Section \ref{ch:hld} to \texttt{basic}.
\item \texttt{vEB}\footnote{For the implementation of van Emde Boas trees,
    we used the following library:
    \url{https://kopricky.github.io/code/Academic/van_emde_boas_tree.html}}: $O(mn \log \log n)$-time and $O(mn)$-space algorithm obtained by combining Algorithm~\ref{alg:dp} in Section~\ref{ch:algobasic} with Algorithm~\ref{alg:vEB} in Section~\ref{ch:algofaster}, and
    \item \texttt{vEB-HL}: $O(mn \log \log n)$-time and $O(n \log m)$-space algorithm obtained by applying the idea of memory reduction in Section \ref{ch:hld} to \texttt{vEB}.
\end{itemize}

Tables~\ref{tab:seq-cons-time} and~\ref{tab:seq-cons-time-worst-case} show the comparison of the performance among four algorithms above.
NA indicates that the measurement was terminated when the execution time exceeded 60 minutes.
Common to both Table~\ref{tab:seq-cons-time} and Table~\ref{tab:seq-cons-time-worst-case},
we use a text $T$ of length $n$ that is a randomly chosen permutation of $(1, 2, \ldots, n)$, and thus, $T$ is a length-$n$ string over the alphabet $\{1, 2, \ldots, n\}$.
In Table~\ref{tab:seq-cons-time}, we use a pattern $P$ that is a randomly chosen subsequence of $T$, and thus, $P$ is also a length-$m$ string over the alphabet $\{1, 2, \ldots, n\}$.
In Table~\ref{tab:seq-cons-time-worst-case}, we use the pattern $P = (k + 1, 1, \ldots, k + i, i, \ldots, 2k, k)$ of length $m = 2k$ in Equation~\ref{eq:worst-case-pattern} (see also Figure~\ref{fig:hld}), which requires $\Theta(mn)$ space when the idea of memory reduction in Section~\ref{ch:hld} is not applied.

\begin{table}[tb]
  \centering
  \caption{Comparison of four algorithms for solving {\CTSM} with randomly generated texts and patterns. The unit of time is second, and the unit of space is KB.}
  \label{tab:seq-cons-time}
  \begin{tabular}{|rr||rr|rr||rr|rr|}
    \hline
    & & \multicolumn{2}{c}{\texttt{basic}} & \multicolumn{2}{c}{\texttt{basic-HL}} & \multicolumn{2}{c}{\texttt{vEB}} & \multicolumn{2}{c|}{\texttt{vEB-HL}} \\ \hline \hline
    $n$ & $m$ & time & space & time & space & time & space & time & space \\ \hline
    $5000$ & $50$ & 2.03 & \textbf{1980} & 0.09 & 3148 & \textbf{0.03} & 2496 & \textbf{0.03} & 2124 \\
    $5000$ & $500$ & 19.20 & 2788 & 19.86 & \textbf{2168} & \textbf{0.37} & 3272 & \textbf{0.37} & 2596 \\
    $5000$ & $1000$ & 40.62 & 2932 & 40.34 & \textbf{2236} & \textbf{0.73} & 3520 & \textbf{0.73} & 2604 \\
    $5000$ & $2500$ & 96.27 & 3124 & 96.23 & \textbf{2368} & \textbf{1.84} & 3532 & \textbf{1.84} & 2816 \\ \hline
    $10000$ & $50$ & 7.77 & 2128 & 7.74 & \textbf{1804} & \textbf{0.07} & 2504 & \textbf{0.07} & 2188 \\
    $10000$ & $1000$ & 159.82 & 2740 & 159.70 & \textbf{1960} & \textbf{1.38} & 3128 & \textbf{1.38} & 2352 \\
    $10000$ & $2000$ & 321.07 & 2920 & 323.09 & \textbf{2068} & \textbf{3.08} & 3312 & 3.09 & 2452 \\
    $10000$ & $5000$ & 841.85 & 3252 & 835.29 & \textbf{2212} & \textbf{7.22} & 3644 & 7.23 & 2592 \\ \hline
    $50000$ & $50$ & 206.49 & 4976 & 211.24 & \textbf{3836} & \textbf{0.39} & 6076 & 0.40 & 4920 \\
    $50000$ & $5000$ & NA & NA & NA & NA & 39.98 & 13040 & \textbf{39.70} & \textbf{6576} \\
    $50000$ & $10000$ & NA & NA & NA & NA & \textbf{79.42} & 12684 & 80.20 & \textbf{7044} \\
    $50000$ & $25000$ & NA & NA & NA & NA & 199.14 & 13900 & \textbf{197.71} & \textbf{7340} \\ \hline
    \end{tabular}
\end{table}

\begin{table}[th]
  \centering
  \caption{Comparison of four algorithms for solving {\CTSM} with randomly generated texts and intentionally generated patterns of form $P=(k + 1, 1, \ldots, k + i, i, \ldots, 2k, k)$ in Equation~\ref{eq:worst-case-pattern}. The unit of time is second, and the unit of space is KB.}
  \label{tab:seq-cons-time-worst-case}
  \begin{tabular}{|rr||rr|rr||rr|rr|}
    \hline
    & & \multicolumn{2}{c}{\texttt{basic}} & \multicolumn{2}{c}{\texttt{basic-HL}} & \multicolumn{2}{c}{\texttt{vEB}} & \multicolumn{2}{c|}{\texttt{vEB-HL}} \\ \hline \hline
    $n$ & $m$ & time & space & time & space & time & space & time & space \\
    \hline
    $5000$ & $50$ & 1.85 & 2572 & 1.86 & \textbf{1940} & \textbf{0.03} & 2920 & \textbf{0.03} & 2208 \\
    $5000$ & $500$ & 18.01 & 11712 & 18.03 & \textbf{1912} & \textbf{0.23} & 12064 & \textbf{0.23} & 2372 \\
    $5000$ & $1000$ & 37.65 & 21804 & 37.94 & \textbf{2028} & 0.41 & 22236 & \textbf{0.40} & 2516 \\
    $5000$ & $2500$ & 92.58 & 52036 & 89.04 & \textbf{2220} & 0.96 & 52720 & \textbf{0.94} & 2960 \\ \hline
    $10000$ & $50$ & 7.39 & 3444 & 7.45 & \textbf{1644} & \textbf{0.07} & 3748 & \textbf{0.07} & 2032 \\
    $10000$ & $1000$ & 150.70 & 41632 & 153.18 & \textbf{1732} & 0.80 & 42192 & \textbf{0.79} & 2304 \\
    $10000$ & $2000$ & 301.57 & 81856 & 303.77 & \textbf{1852} & 1.49 & 82584 & \textbf{1.46} & 2600 \\
    $10000$ & $5000$ & 754.85 & 202408 & 759.71 & \textbf{2244} & 3.58 & 203656 & \textbf{3.49} & 3512 \\ \hline
    $50000$ & $50$ & 186.05 & 12024 & 186.63 & \textbf{3048} & \textbf{0.37} & 13116 & \textbf{0.37} & 4140 \\
    $50000$ & $5000$ & NA & NA & NA & NA & 18.36 & 650768 & \textbf{17.82} & \textbf{5616} \\
    $50000$ & $10000$ & NA & NA & NA & NA & 35.42 & 963068 & \textbf{34.25} & \textbf{7112} \\
    $50000$ & $25000$ & NA & NA & NA & NA & 87.28 & 998056 & \textbf{83.94} & \textbf{11600} \\ \hline
    \end{tabular}
\end{table}

Table~\ref{tab:seq-cons-time} shows that the running time of \texttt{vEB} is faster than that of \texttt{basic} for all test cases,
and the same result can be seen for \texttt{vEB-HL} and \texttt{basic-HL}.
Comparing the memory usage of \texttt{vEB} with that of \texttt{basic}, it can be seen that the \texttt{vEB} uses more memory than \texttt{basic},
since the memory usage of the van Emde Boas tree is constant times larger than that of a basic array.
The same is true for \texttt{vEB-HL} and \texttt{basic-HL}.
The only difference between \texttt{basic} (\texttt{vEB}) and \texttt{basic-HL} (\texttt{vEB-HL}) is the search order of the tree traversal, so they have little difference in the running time for all test cases.
Comparing these algorithms in terms of memory usage,
it can be seen that the \texttt{basic-HL} (\texttt{vEB-HL}) uses less memory than \texttt{basic} (\texttt{vEB}),
but the difference is not as pronounced as the theoretical difference in the space complexity.
This is because $P$ is generated at random, so there is not much bias in the size of the subtrees.

On the other hand, the results in Table~\ref{tab:seq-cons-time-worst-case} show that \texttt{basic-HL} and \texttt{vEB-HL} are significantly more memory efficient than \texttt{basic} and \texttt{vEB} in the case where $m$ is large.
This is consistent with the theoretical difference in the amount of the space complexity.

We also conducted the additional experiments with other algorithms:
\begin{itemize}
    \item \texttt{BST}: $O(mn \log n)$-time and $O(mn)$-space algorithm using the \emph{binary search tree}\footnote{For the implementation of binary search trees, we used \texttt{std::set} in C++.} instead of van Emde Boas tree in Section~\ref{ch:algofaster}, and
    \item \texttt{BST-HL}: $O(mn \log n)$-time and $O(mn)$-space algorithm obtained by applying the idea of memory reduction in Section~\ref{ch:hld} to \texttt{BST}.
\end{itemize}
\texttt{vEB} outperformed \texttt{BST} in both time and space for all test cases, and so do \texttt{vEB-HL} and \texttt{BST-HL}, which we feel is of independent interest.
The details of the results are shown in Appendix~\ref{app:table}.

 \section{Conclusions}
This paper introduced the Cartesian tree subsequence matching ({\CTSM}) problem: Given a text $T$ of length $n$ and a pattern $P$ of length $m$, find every minimal substring $S$ of $T$ such that $S$ contains a subsequence $S'$ which Cartesian-tree matches $P$.
This is the Cartesian-tree version of the episode matching~\cite{DasFGGK97}.
We first presented a basic dynamic programming algorithm running in $O(mn^2)$ time, and then proposed a faster $O(mn \log \log n)$-time solution to the problem. 
We showed how these algorithms can be performed with $O(n \log m)$ space.
Our experiments showed that our $O(mn \log \log n)$-time solution can be fast in practice.

An intriguing open problem is to show a non-trivial (conditional) lower bound for the \CTSM\ problem.
The episode matching (under the exact matching criterion) has $O((mn)^{1-\epsilon})$-time conditional lower bound under SETH~\cite{abs-2108-08613}.
Although a solution to the \CTSM\ problem that is significantly faster than $O(mn)$ seems unlikely, we have not found such a (conditional) lower bound yet.
We remark that the episode matching problem is not readily reducible to the \CTSM\ problem, 
since \CTSM\ allows for more relaxed pattern matching and the reported intervals can be shorter than those found by episode matching.

\subsection*{Acknowledgments}
This work was supported by JSPS KAKENHI Grant Numbers JP20J11983 (TM), 20H00595 (HA),
JST PRESTO Grant Number JPMJPR1922 (SI), and JST CREST Grant Number JPMJCR18K3 (HA).

The authors thank the anonymous referees for drawing our attention to reference~\cite{GawrychowskiGL20}.

\clearpage
\appendix
\section{Additional Table}\label{app:table}

\begin{table}[h!]
  \centering
  \caption{Comparison of six algorithms with additional two algorithms for solving {\CTSM} with randomly generated texts and patterns. The unit of time is second, and the unit of space is KB.}
  \label{tab:seq-cons-time-all}
    \rotatebox{90}{
        \begin{tabular}{|rr||rr|rr||rr|rr||rr|rr|}
        \hline
        & & \multicolumn{2}{c}{\texttt{basic}} & \multicolumn{2}{c}{\texttt{basic-HL}} & \multicolumn{2}{c}{\texttt{BST}} & \multicolumn{2}{c}{\texttt{BST-HL}} & \multicolumn{2}{c}{\texttt{vEB}} & \multicolumn{2}{c|}{\texttt{vEB-HL}} \\
        \hline \hline
        $n$ & $m$ & time & space & time & space & time & space & time & space & time & space & time & space \\ \hline
        $5000$ & $50$ & 2.03 & \textbf{1980} & 2.03 & 2020 & 0.09 & 3284 & 0.09 & 3148 & \textbf{0.03} & 2496 & \textbf{0.03} & 2124 \\
        $5000$ & $500$ & 19.20 & 2788 & 19.86 & \textbf{2168} & 0.85 & 3896 & 0.83 & 3240 & \textbf{0.37} & 3272 & \textbf{0.37} & 2596 \\
        $5000$ & $1000$ & 40.62 & 2932 & 40.34 & \textbf{2236} & 1.68 & 4084 & 1.67 & 3348 & \textbf{0.73} & 3520 & \textbf{0.73} & 2604 \\
        $5000$ & $2500$ & 96.27 & 3124 & 96.23 & \textbf{2368} & 4.21 & 4396 & 4.18 & 3480 & \textbf{1.84} & 3532 & \textbf{1.84} & 2816 \\ \hline
        $10000$ & $50$ & 7.77 & 2128 & 7.74 & \textbf{1804} & 0.20 & 4076 & 0.19 & 3360 & \textbf{0.07} & 2504 & \textbf{0.07} & 2188 \\
        $10000$ & $1000$ & 159.82 & 2740 & 159.70 & \textbf{1960} & 3.70 & 4724 & 3.64 & 3940 & \textbf{1.38} & 3128 & \textbf{1.38} & 2352 \\
        $10000$ & $2000$ & 321.07 & 2920 & 323.09 & \textbf{2068} & 8.25 & 4912 & 8.22 & 4048 & \textbf{3.08} & 3312 & 3.09 & 2452 \\
        $10000$ & $5000$ & 841.85 & 3252 & 835.29 & \textbf{2212} & 20.25 & 5232 & 19.69 & 4196 & \textbf{7.22} & 3644 & 7.23 & 2592 \\ \hline
        $50000$ & $50$ & 206.49 & 4976 & 211.24 & \textbf{3836} & 1.46 & 10204 & 1.45 & 10004 & \textbf{0.39} & 6076 & 0.40 & 4920 \\
        $50000$ & $5000$ & NA & NA & NA & NA & 141.22 & 17276 & 136.76 & 10868 & 39.98 & 13040 & \textbf{39.70} & \textbf{6576} \\
        $50000$ & $10000$ & NA & NA & NA & NA & 271.18 & 16920 & 272.29 & 11440 & \textbf{79.42} & 12684 & 80.20 & \textbf{7044} \\
        $50000$ & $25000$ & NA & NA & NA & NA & 691.63 & 18144 & 689.80 & 11780 & 199.14 & 13900 & \textbf{197.71} & \textbf{7340} \\ \hline
        \end{tabular}
    }
\end{table}

\end{document}